\documentclass[11pt,a4paper,reqno]{amsart}

\usepackage{hyperref}
\usepackage[foot]{amsaddr}
\usepackage{changepage}
\usepackage{amsmath}
\usepackage{amsthm}
\usepackage{amsfonts}
\usepackage{amssymb}
\usepackage{array}
\usepackage[margin=2.4cm]{geometry}

\usepackage{soul}

\usepackage{enumitem}
\makeatletter
\newcommand{\mylabel}[2]{#2\def\@currentlabel{#2}\label{#1}}
\makeatother

\theoremstyle{plain}
\newtheorem{theorem}{Theorem}
\newtheorem*{theorem*}{Theorem}
\newtheorem{proposition}[theorem]{Proposition}
\newtheorem{corollary}[theorem]{Corollary}
\newtheorem*{corollary*}{Corollary}
\newtheorem{lemma}[theorem]{Lemma}

\newtheorem{problem}{Problem}

\theoremstyle{remark}
\newtheorem{remark}[theorem]{Remark}
\newtheorem{example}[theorem]{Example}

\theoremstyle{definition}
\newtheorem{definition}[theorem]{Definition}

\usepackage{tikz}

\newcommand{\Zset}{\mathbb{Z}}
\newcommand{\F}{\mathbb{F}}
\newcommand{\field}[1][]{\mathbb{F}_{#1}}
\DeclareMathOperator{\weight}{w}

\DeclareMathOperator{\distance}{d}
\DeclareMathOperator{\rank}{rk}

\DeclareMathOperator{\trace}{Tr}
\DeclareMathOperator{\GL}{GL}
\DeclareMathOperator{\Gal}{Gal}

\newcommand{\matrixring}[2]{#2^{#1}}
\newcommand{\norm}[2]{N_{#1}(#2)}
\newcommand{\lclm}[1]{\left[#1\right]_\ell}

\newcommand{\npmatrix}[1]{\left( \begin{matrix} #1 \end{matrix} \right)}

\newcommand{\B}{\mathcal B}
\newcommand{\C}{\mathcal C}
\newcommand{\G}{\mathcal G}

\begin{document}

\title{
Roos bound for skew cyclic codes in Hamming and rank metric
}

\author[Gianira N. Alfarano]{Gianira N. Alfarano$^1$}
\address{$^1$University of Zurich, Switzerland}
\curraddr{}
\email{gianiranicoletta.alfarano@math.uzh.ch}

\author[F. J. Lobillo]{F. J. Lobillo$^2$}
\address{$^2$University of Granada, Spain}
\curraddr{}
\email{jlobillo@ugr.es}

\author[Alessandro Neri]{Alessandro Neri$^3$}
\address{$^3$Technical University of Munich, Germany}
\curraddr{}
\email{alessandro.neri@tum.de}
\thanks{Research partially supported by grant PID2019-110525GB-I00
from {Agencia Estatal de Investigaci\'{o}n (AEI) and from Fondo Europeo de Desarrollo Regional (FEDER)}, and from Swiss National Science Foundation through grants no. 187711 and 188430.}



\maketitle

\begin{abstract}
In this paper, a Roos like bound on the minimum distance for skew cyclic codes over a general field is provided. The result holds in the Hamming metric and in the rank metric. The proofs involve arithmetic properties of skew polynomials and an analysis of the rank of parity-check matrices. For the rank metric case, a way to arithmetically construct codes with a prescribed minimum rank distance, using the skew Roos bound, is also given. Moreover, some examples of MDS codes and MRD codes over finite fields are built, using the skew Roos bound. 
\end{abstract}

\section{Introduction}
In the theory of error correcting codes, a crucial step was represented by the introduction of algebraic structures, which led to the branch called \emph{algebraic coding theory}. More specifically,  the basic idea that initiated the study of \emph{linear codes} was to take a finite field $F$ as alphabet, and then use the vector space structure $F^n$ when dealing with codes and their codewords. Among the linear codes, one of the most studied families is the one of \emph{cyclic codes}. Their importance is given by the ring structure underlying  their polynomial representation. Formally, 
a cyclic block code $\mathcal C$ over $F$ is defined as an ideal of $F[x]/(x^n-1)$. It is well-known that the minimum Hamming distance of a cyclic code is lower bounded by the BCH bound \cite{bose1960class, bose1960further,hocquenghem1959codes}. Concretely, let $g(x)$ be the generator polynomial of $\mathcal C$, $\omega$ be a primitive $n$-th root of unity and $b,\delta$ be positive integers. If $g(\omega^{b+i})=0$ for $0\leq i\leq \delta-2$, i.e. $g$ has $\delta-1$ consecutive roots in an extension field of $F$, then the minimum Hamming distance of $\mathcal C$ is at least $\delta$. The cyclicity property was further investigated in order to construct codes with prescribed  Hamming distance. At a second step, Hartmann and Tzeng generalized the BCH bound deriving the well-known Hartmann-Tzeng (HT) bound \cite{hartmann1972generalizations}. Afterwards, Roos  derived further generalizations  which were shown to improve both the BCH and the HT bounds  \cite{roos1983new, roos1982generalization}.

Skew polynomial rings were introduced in 1930 by Ore in his seminal paper \cite{ore1933theory} and then they have been further studied by several authors, see for instance \cite{lam1985general,lam1988vandermonde, leroy2012noncommutative}.
The research on linear codes in this setting led to new codes with better parameters than the old known linear codes.  In 2007, Boucher, Geiselmann and Ulmer \cite{boucher2007skew} extended the definition of cyclicity to codes defined over the skew polynomial ring (see also \cite{boucher2009codes, boucher2009coding, chaussade2009skew}).
In these works, the authors derived bounds on the Hamming distance of skew cyclic codes, generalizing in some sense the BCH bound to skew cyclic codes. In \cite{Gomez/Lobillo/Navarro/Neri:2018}, the authors gave a version of the Hartmann-Tzeng bound for skew cyclic codes and provided a construction of these codes with prescribed designed Hamming distance. 

Skew polynomial rings played a crucial role also in the construction of codes endowed with the rank metric. These codes were first introduced independently by Delsarte \cite{delsarte1978}, Gabidulin \cite{gabidulin1985} and Roth \cite{roth1991}, and have been shown to have many applications, such as network coding \cite{silva2008rank, silva2011universal, etzion2018vector}, distributed data storage \cite{rawat2013optimal, calis2017general, ne18pmds} and  post-quantum cryptography \cite{gabidulin1991, overbeck2008structural, gaborit2013low}. One of the most important constructions of rank-metric codes makes use of the ring of linearized polynomials. More specifically, these codes are known as \emph{Gabidulin codes} and they are obtained by evaluating a particular set of linearized polynomials in a suitable set of points (see \cite{delsarte1978, gabidulin1985}). The connection with skew polynomials is due to the fact that there is a natural isomorphism between the ring of linearized polynomials over a finite field and the ring of skew polynomials. Generalizations of this construction were provided in \cite{chaussade2009skew}, where a rank-metric version of the BCH bound was proposed. Moreover the analogue of the Hartmann-Tzeng bound for skew-cyclic codes over finite field with respect to the rank metric was shown in \cite{martinez2017roots}. 

In this paper we provide a generalization of the Roos bound for skew cyclic codes in the Hamming metric and in the rank metric. Our results generalize previous bounds on the minimum distance of skew cyclic codes in the Hamming metric \cite{boucher2009codes, chaussade2009skew, Gomez/Lobillo/Navarro/Neri:2018}, and in the rank metric \cite{delsarte1978, gabidulin1991, roth1991, chaussade2009skew, augot2013rank, martinez2017roots}. However, our setting only requires a cyclic Galois extension of finite degree, without restricting  to the case of finite fields.

The paper is structured as follows. In Section \ref{sec:scc} we recall the basics of skew polynomial rings and the notions of skew cyclic codes, focusing on the family of skew Reed-Solomon
codes. Section \ref{sec:rankmetric} is dedicated to the rank metric. We define the rank metric in the most general setting and describe the construction of Gabidulin codes over any cyclic Galois extension. In Section \ref{sec:defsets} we fix the mathematical setting for the whole paper, focusing on the defining sets for skew cyclic codes, and we prove the results that are crucial for the main proofs. Section \ref{sec:RboundHamming} is devoted to the proof of the skew version of the Roos bound for the Hamming metric. We use the bound to construct some examples of (MDS) codes over finite fields. In Section \ref{sec:RboundRank} we provide the skew version of the Roos bound for the rank metric. We compare the construction of the codes in the Hamming metric with the one in the rank metric, which led to an interesting result, explaining that it could be possible to construct skew cyclic MRD codes, using the arithmetic properties of the defining sets. We conclude with some remarks and an open problem in Section \ref{sec:conclusion}.


\section{Skew cyclic codes and skew Reed-Solomon codes}\label{sec:scc}

In this section we recall some basic notions on skew polynomial rings and skew cyclic codes. The interested reader is referred to the recent survey of Gluesing-Luerssen \cite{gluesing2019skew}.

We will use the notation introduced in \cite[\S 2]{Gomez/Lobillo/Navarro/Neri:2018} to recall the definition of skew cyclic codes and some important known results. Let $F/K$ be a field extension of finite degree \(\mu\). We assume \(F/K\) is cyclic, i.e. its Galois group, \(\Gal(F/K)\), is cyclic. Fix a  generator \(\sigma\) of $\Gal(F/K)$, hence its order \(|\sigma|\) is \(\mu\) and \(K = F^\sigma\) is its invariant subfield. Let \(R = F[x;\sigma]\) be the skew polynomial ring induced by $\sigma$ over $F$ and $n$ be a multiple of $\mu$, namely \(n = \nu \mu\) for some $\nu$ positive integer. Recall that the multiplication rule over the skew polynomial ring $R$ is given by
\[ xa = \sigma(a)x  \text{ for all } a\in F .\]
 In order to define skew cyclic codes over $F$ it is enough to replace $F[x]/(x^n-1)$ by $\mathcal{R} := R/R(x^n-1)$, where $R(x^n-1)$ denotes the left ideal generated by $x^n-1$. Since $\sigma$ has finite order \(\mu\), the center $Z(R)$ of $R$ is given by the commutative polynomial ring $K[x^\mu]$ and therefore $x^n-1$ belongs to $Z(R)$. Hence, \(R(x^n-1)\) is a twosided ideal and the quotient $\mathcal R$ is a $K$-algebra which is isomorphic to $F^n$ as $F$-vector space, where the \(F\) action is given by left multiplication. 
 
 An $[n,k]$-\emph{linear code} $\C$ over $F$ is defined as a subspace of $F^n$ of dimension $k$. Hence, thanks to isomorphism described above, we can identify linear codes in $F^n$ as vector subspaces of $\mathcal R$. We define the \emph{Hamming distance} between two vectors in $F^n$ as the number of components in which they differ.  The \emph{minimum (Hamming) distance} of a linear code $\C$ is defined as the minimum over all the  distances between two distinct codewords in $\C$ and it is denoted by $\distance_H(\C)$. Equivalently, $\distance_H(\C)$ is given by the minimum (Hamming) weight of the nonzero codewords in $\C$, where the \emph{Hamming weight} of a vector $v\in F^n$ is defined as the number of its nonzero components  and it is denoted by $\weight(v)$. When the minimum distance $d=\distance_H(\C)$ is known, we write that $\C$ is an $[n,k,d]$-linear code. The parameters $n,k$ and $d$ of a linear code $\C$ satisfy the following inequality, known as Singleton bound \cite{singleton1964maximum}: $d\leq n-k+1$. When the minimum distance of $\C$ meets the bound with the equality, $\C$ is called \emph{maximum distance separable (MDS) code}.
 
 In the setting defined above, an \([n,k]\)-linear code \(\mathcal{C} \subseteq F^n\) is called \emph{skew cyclic} if $\mathcal{C}$ is a left ideal of $\mathcal R$. 

Since every left ideal of $\mathcal R$ is principal, there exists a polynomial \(g \in R\) which is a right divisor of $x^n-1$, namely \(g \mid_r x^n-1\), such that  \(\deg(g) = n-k\) and $g$ generates $\mathcal C$. We will write \(\mathcal{C} = \mathcal{R}g\).  


Evaluation in skew polynomials makes use of truncated norms.
%
For any $i\in \mathbb{N}_0$, the \emph{$i$-th truncated norm on $F$} is defined as
\(N_i: F\to F,\)
with $N_0(a) = 1$ and $N_i(a) =\prod_{j=0}^{i-1}\sigma^j(a)$ for $i>0$, for any $a\in F$. This is a special case of \cite[(2.3)]{lam1988vandermonde}, where a deep discussion of evaluations can be found.
%
Note that $N_1(a)=a$ and $N_{i+1}(a) = N_i(a)\sigma^i(a)$ for any $i>0$. 
If $f(x)=\sum_{i=0}^r f_ix^i \in R$, it follows that 
\[
f(x) = q(x) (x-a) + \sum_{i=0}^r f_i N_i(a)
\]
as proved in \cite[Lemma 2.4]{lam1988vandermonde}, hence $f(a) = \sum_{i=0}^r f_i N_i(a)$ is the correct notion of evaluation of skew polynomials. 

Since \(R\) is a left (and right) Euclidean domain, greatest common right divisor and least common left multiple exist and can be computed. We denote the least common left multiple of two polynomials \(f,g \in R\) by \(\lclm{f,g}\). A detailed  computational treatment of skew polynomials can be found in \cite{Gomez:2013}.

The structure of a skew cyclic code
is better understood if a full decomposition of \(g\) as least common left multiple of linear polynomials can be provided.
Let \(E/K\) be a cyclic field extension of degree \(n\) and \(\theta \in \Gal(E/K)\) an automorphism of degree \(n\), i.e. \(K = E^\theta\). Let \(S = E[x;\theta]\) and \(\mathcal{S} := S/S(x^n-1)\). Recall that \(\mathcal{S} \cong \matrixring{n \times n}{K}\) is simple Artinian, see for instance \cite[Theorem 1]{Gomez/Lobillo/Navarro:2016b}. Hence all simple modules are isomorphic and given \(g \in S\) of degree $n-k$ with \(g \mid_r x^n-1\), there exist \(\beta_0, \beta_1, \dots, \beta_{n-k-1} \in E\) such that $g$ is the least common left multiple of $\{x - \beta_i \mid 0\leq i\leq n-k-1\}$, that is
\[
g = \lclm{x - \beta_0, \dots, x-\beta_{n-k-1}}.
\]

Recall that \(x - \beta \mid_r x^n-1\) if and only if \(\norm{n}{\beta} = 1\), and by Hilbert's Theorem 90 (see e.g. \cite[Chapter VI, Theorem 6.1]{Lang:2002}) this happens if and only if \(\beta = \theta(\alpha) \alpha^{-1}\) for some \(\alpha \in E\). Hence \(\beta_i = \theta(\alpha_i) \alpha_i^{-1}\) for \(K\)-linear independent \(\alpha_0, \dots, \alpha_{n-k-1} \in E\), see \cite[Theorem 5.3]{Delenclos/Leroy:2007}. When these linear independent elements are part of a normal basis, a better knowledge of the parameters of the code is obtained. Concretely, we have the following result. 


\begin{proposition}[{\cite[Theorem 4]{Gomez/Lobillo/Navarro:2017c}}]\label{prop:skewBCHbound}
Let \(\alpha \in E\) such that \(\left\{\alpha, \theta(\alpha), \dots, \theta^{n-1}(\alpha)\right\}\) is a normal basis and \(\beta = \theta(\alpha) \alpha^{-1}\). Let $1\leq \delta \leq n$ and  \(g = \lclm{\left\{ x - \theta^i(\beta) ~|~ 0 \leq i \leq \delta-2 \right\}}\). Then \(\mathcal{S}g \subseteq \mathcal{S}\) is an MDS code of length \(n\) and minimum Hamming distance \(\delta\). 
\end{proposition}

These codes are usually called \emph{skew Reed-Solomon codes}, and denoted by 
\[
\operatorname{sRS}^\theta_{\beta}(n,\delta) = \mathcal{S} \lclm{\left\{ x - \theta^i(\beta) ~|~ 0 \leq i \leq \delta-2 \right\}}.
\]

From Proposition \ref{prop:skewBCHbound}, one can easily recognize the generalization of the BCH bound to skew cyclic codes.

The proof of Proposition \ref{prop:skewBCHbound} is based on Circulant Lemma, which is a particular case of \cite[Corollary 4.13]{lam1988vandermonde}. We include the statement since it is going to be used to prove the main results of this paper. 
  
\begin{lemma}[Circulant Lemma]\label{circulantlemma}
Let \(\{\alpha_0, \dots, \alpha_{n-1}\}\) be a \(K\)--basis of \(E\). Then, for every positive integer \(t \leq n\) and every subset \(\{k_1, k_2, \dots, k_{t}\} \subseteq \{0, 1, \dots, n-1\}\),
\[
\begin{vmatrix}
\alpha_{k_1} & \theta(\alpha_{k_1}) & \dots & \theta^{t-1}(\alpha_{k_1}) \\
\alpha_{k_2} & \theta(\alpha_{k_2}) & \dots & \theta^{t-1}(\alpha_{k_2}) \\
\vdots & \vdots &  & \vdots \\
\alpha_{k_{t}} & \theta(\alpha_{k_{t}}) & \dots & \theta^{t-1}(\alpha_{k_{t}})
\end{vmatrix} \neq 0.
\]
\end{lemma}

An elementary proof is available in \cite{Gomez/Lobillo/Navarro:2017c}.

\section{Rank metric over any field extension}\label{sec:rankmetric}

Although we will always consider cyclic extensions, here we discuss the rank metric in the most general case, in the spirit of the recent works of Augot, Loidreau and Robert \cite{augot2013rank,augot2014generalization, augot2018generalized}. A similar approach was first investigated by Roth in \cite[Section 6]{roth1996tensor}. 


Let $K$ be a field and $E$ be an extension field of degree $n$. Let $\theta \in \Gal(E/K)$ with order \(|\theta| = \eta\), so $\eta$ divides $n$. Let moreover $\B=\{b_1,\ldots, b_n\}$ be an ordered $K$-basis of $E$. For a given vector $v=(v_1,\ldots, v_N) \in E^N$, we consider the following two matrices:
\[
M_{v,\theta}:=\npmatrix{v_1 & v_2 & \cdots & v_N \\ \theta(v_1) & \theta(v_2) & \cdots & \theta(v_N) \\
\vdots & \vdots &  & \vdots \\  \theta^{\eta-1}(v_1) & \theta^{\eta-1}(v_2) & \cdots & \theta^{\eta-1}(v_N) }, \]
\[
M_{v,\B}:=\npmatrix{x_{1,1} & x_{2,1} & \cdots & x_{N,1} \\  x_{1,2} & x_{2,2} & \cdots & x_{N,2} \\  \vdots & \vdots & & \vdots \\
x_{1,n} & x_{2,n} & \cdots & x_{N,n}},\]
where $v_i=\sum_{j=1}^n x_{i,j}b_j$, for every $i=1,\ldots, N$.

Augot, Loidreau and Robert  defined in \cite{augot2013rank} two different rank-weights for a vector $v \in E^N$ as follows. Let $E$, $K$ and $\theta$ be as above, and let $v \in E^N$. The quantities $\weight_K(v)$ and $\weight_E(v)$ are defined as
\begin{align*}\weight_K(v):= & \rank_K(M_{v,\theta})=\rank_K(M_{v,\B}), \\
\weight_E(v):= & \rank_E(M_{v,\theta})=\deg(p_v),\end{align*}
where $p_v=\lclm{x-v_1, \ldots, x-v_N} \in E[x;\theta]$. It was shown by the same authors that these quantities are all equal in a special case.

\begin{proposition}\cite[Proposition 5]{augot2013rank}
If $K=E^\theta$, then $\weight_E(v)=\weight_K(v)$, and they are both equal to
\[ \weight_R(v):=\dim_{K}\langle v_1,\ldots, v_N\rangle_K.\]
\end{proposition}
%
%

For the rest of this section we will only deal with the case $E^\theta=K$, i.e. \(\Gal(E/K)\) is cyclic and the order of \(\theta\) is \(n\), therefore we will use the notation $\weight_R(v)$ to denote the \emph{rank weight} of a vector with respect to the cyclic extension $E/K$.

\begin{definition}
Let $E/K$ be a cyclic field extension of finite degree, then the \emph{rank distance} of two vectors $u,v \in E^N$ with respect to the extension $E/K$  is defined as  $\distance_R(u,v):=\weight_R(u-v)$.
\end{definition}

With this metric, we can introduce the notion of rank-metric codes.

\begin{definition}
Let $E/K$ be a cyclic finite extension field and let $N, k, d$ be positive integers. An $[N,k,d]_{E/K}$ \emph{rank-metric code} $\C$ is a $k$-dimensional $E$-subspace of $E^N$, endowed with the rank metric. The integer $N$ is called the \emph{length} of $\C$, $k$ is the \emph{dimension} of $\C$ and  $d$ is defined as
\[d=\distance_R(\C):=\min \{ \distance_R(u,v) \mid u,v \in \C, u \neq v \}\]
and is called \emph{minimum rank distance} of $\C$.
\end{definition}

\begin{definition}
 Let $k \leq N$ and $g \in E^N$ be a vector such that $\weight_R(g)=N$, and $\tau$ be a generator of $\Gal(E/K)$. Then, the \emph{$\tau$-Gabidulin code} $\G_{k,\tau}(g)$ is the code
\[\G_{k,\tau}(g)=\left\langle g, \tau(g), \ldots, \tau^{k-1}(g)\right\rangle.\]
\end{definition}

Observe that in the definition we are implicitly assuming that $n \geq N$, since for every $v \in E$ we have $\weight_R(v) \leq [E:K]=n$.

Gabidulin codes were constructed by Delsarte \cite{delsarte1978} and Gabidulin \cite{gabidulin1985} independently over finite fields, when $\tau$ is the Frobenius automorphism, and then generalized by Kshevetsky and Gabidulin in \cite{kshevetskiy2005} to any generator of the Galois group. The general definitions for arbitrary fields were  due to Roth in \cite[Section 6]{roth1996tensor} and to Augot, Loidreau and Robert in \cite{augot2013rank}. 

These codes are known to be \emph{maximum rank distance (MRD)}, i.e. the minimum rank distance of a Gabidulin code is $N-k+1$, which is the maximum possible value according to the Singleton-like bound for the rank metric (see \cite{delsarte1978, gabidulin1985} for the finite field case, \cite{augot2013rank} for general fields). Moreover, it is well-known that Gabidulin codes are closed under duality. This means that for every $g \in E^N$ such that $\weight_R(g)=N$,  there exists $h \in E^N$ such that $\weight_R(h)=N$ and $\G_{k,\tau}(g)^\perp=\G_{N-k,\tau}(h)$, where the dual is taken with respect to the standard inner product.

There is a way to characterize the minimum rank distance of an $[N,k,d]_{E/K}$ rank-metric code $\C$ in terms of the minimum Hamming distance of a family of linear block codes obtained from $\C$. This is explained by the next proposition, which directly follows from \cite[Theorem 1]{gabidulin1985} for the finite field case. For this purpose, we introduce the following notation. For a given set $\mathcal H \subseteq E^N$, and a matrix $M \in K^{N \times N}$, we write $\mathcal H \cdot M:=\{uM \mid u \in \mathcal H\}$.

\begin{proposition}\label{prop:minrank}
 Let $\C$ be an $[N,k,d]_{E/K}$ rank-metric code. Then
 $$\distance_R(\C)=\min\left\{ \distance_H(\C \cdot M) \mid M \in \GL_N(K) \right\}.$$
\end{proposition}

\begin{proof}
 Let $\delta:=\min\left\{ \distance_H(\C \cdot M) \mid M \in \GL_N(K) \right\}$ and $d:=\distance_R(\C)$. For every $c \in \C$, $M \in \GL_N(K)$, we have $\weight_R(cM)=\weight_R(c)$, and $\weight_R(cM)\leq \weight(cM)$. Hence, $\delta \geq d$. On the other hand, suppose that $c\in \C$ is of minimal rank-weight. Let $S:=\langle c_1,\ldots, c_N \rangle_{K}$ that for hypothesis has dimension $d$ over $K$, and choose a basis $v_1,\ldots,v_d$ of $S$. Hence, there exists a matrix $\bar{M} \in \GL_N(K)$ such that $c\bar{M}=(v_1,\ldots, v_d,0,\ldots,0)$. This implies that $\delta\leq \distance_H(\C \cdot \bar{M})\leq \weight(c\bar{M}) = d$, which concludes the proof.
\end{proof}

\section{Defining sets}\label{sec:defsets}

For the rest of the paper, \(F/K\) will denote an arbitrary cyclic field extension and \(\sigma \in \Gal(F/K)\) an automorphism of order \(|\sigma| = \mu\) such that \(K = F^\sigma\). We say that \(\sigma\) \emph{has an extension \(\theta\) of degree \(\nu\)} if  there exists a field extension \(E/F\) and \(\theta \in \Gal(E/K)\) such that \(|\theta| = n = \nu \mu\), \(\theta_{|F} = \sigma\) and \(E^\theta = F^\sigma = K\). 

We fix such an extension \(E/F\) of degree \(\nu\). 

\begin{center}

\begin{tikzpicture}[node distance = 1cm, auto]
      \node (K) {$K$};
      \node (F) [above of=K] {$F$};
      \node (E) [above of=F] {$E$};
      \draw[-] (K) to node {$\mu$} (F);
      \draw[-] (F) to node {$\nu$} (E);
      \draw[-, bend right] (K) to node [swap] {$n$} (E);
      \end{tikzpicture}
    
\end{center}

Recall that \(R = F[x;\sigma]\), \(\mathcal{R} = \frac{R}{R(x^n-1)}\), \(S = E[x;\theta]\) and  \(\mathcal{S} = \frac{S}{S(x^n-1)}\).  Since for any $f\in R$, we have $Sf\cap R = Rf$ (see \cite[Lemma 2.3]{Gomez/Lobillo/Navarro/Neri:2018}), 
there is a natural inclusion $\mathcal R \subseteq \mathcal S$. As we have observed in Section \ref{sec:scc}, we get that \(\mathcal{S} \cong \matrixring{n \times n}{K}\) as $K$-algebra.

Let \(\mathcal{C} = \mathcal{R} g\) be an \([n,k]\) skew cyclic code with \(g \mid_r x^n-1\), and \(\widehat{\mathcal{C}} = \mathcal{S}g\). It follows that \(\mathcal{C}\) is a subfield subcode of \(\widehat{\mathcal{C}}\).  Moreover, there exist \(\beta_0, \dots, \beta_{n-k-1} \in E\) such that 
\[
g = \lclm{x-\beta_0, \dots, x-\beta_{n-k-1}},
\]
as explained in Section \ref{sec:scc}.

Given \(\{a_0, \dots, a_{t-1}\} \subseteq E\), define the following $n \times t$ matrix:
$$N(a_0, \dots, a_{t-1}) = \Big( \norm{i}{a_j} \Big)_{\substack{0 \leq i \leq n-1 \\ 0 \leq j \leq t-1}} 
= \left( \begin{matrix}
1 & 1 & \cdots & 1 \\
a_0 & a_1 & \cdots & a_{t-1} \\
\norm{2}{a_0} & \norm{2}{a_1} & \cdots & \norm{2}{a_{t-1}} \\
\vdots & \vdots & \ddots & \vdots \\
\norm{n-1}{a_0} & \norm{n-1}{a_1} & \cdots & \norm{n-1}{a_{t-1}}
\end{matrix}\right).
$$

For any matrix \(M\) we denote by \(\ker(M)\) its left kernel, i.e. \(\ker(M) = \{v \mid vM = 0\}\).

\begin{proposition}\label{prop:paritycheck}
Let $\widehat{\mathcal{C}}\subseteq \mathcal S$ be the $[n,k]$ skew cyclic code generated be 
$g=\lclm{x-\beta_0, \dots, x-\beta_{n-k-1}}$, with $\beta_0,\dots, \beta_{n-k-1}\in E$. Then,  \(\widehat{\mathcal{C}} = \ker \left(N(\beta_0, \dots, \beta_{n-k-1})\right)\), i.e. \(N(\beta_0, \dots, \beta_{n-k-1})\) is a parity check matrix for \(\mathcal{C}\) and \(\widehat{\mathcal{C}}\). 
\end{proposition}

\begin{proof}
A polynomial \(f = \sum_{i=0}^{n-1} f_i x^i\) is in \(\widehat{\mathcal{C}}\) if and only if \(x - \beta_j \mid_r f\) for all \(0 \leq j \leq n-k-1\). Since \(x - \beta_j \mid_r f\) if and only if \(\sum_{i=0}^{n-1} f_i \norm{i}{\beta_j}\)=0, the result follows. 
\end{proof}


As we pointed out before, \(x - \beta \mid_r x^n-1\) if and only if \(\beta = \theta(\alpha) \alpha^{-1}\). For all \(\alpha \in E\) we use the notation
\[
\alpha^{[\theta]} = \left( \alpha, \theta(\alpha), \ldots, \theta^{n-1}(\alpha) \right)^\top.
\]

\begin{proposition}\label{prop:kerSRS}
Assume \(\beta_i = \theta(\alpha_i) \alpha_i^{-1}\) for each \(0 \leq i \leq n-k-1\) and let $g=\lclm{x-\beta_0, \dots, x-\beta_{n-k-1}}$. Then \[\widehat{\mathcal{C}} = \mathcal{S}g = \ker \left( \alpha_0^{[\theta]} \big| \alpha_1^{[\theta]} \big| \cdots \big| \alpha_{n-k-1}^{[\theta]} \right). \] 
\end{proposition}

\begin{proof}
Since \(\norm{i}{\beta_j} = \theta^{i}(\alpha_j) \alpha_j^{-1}\), it follows that
\(
\left( 1, \beta_j, \norm{2}{\beta_j}, \cdots, \norm{n-1}{\beta_j} \right)^\top = \alpha_j^{[\theta]} \alpha_j^{-1},
\)
hence
$$N(\beta_0, \dots, \beta_{n-k-1}) = 
\left( \alpha_0^{[\theta]} \big| \alpha_1^{[\theta]} \big| \cdots \big| \alpha_{n-k-1}^{[\theta]} \right) \left(\begin{smallmatrix} \alpha_0^{-1} & & & \\ & \alpha_1^{-1} & & \\ & & \ddots & \\ & & & \alpha_{n-k-1}^{-1}\end{smallmatrix}\right)
$$
and \[\ker \left(N(\beta_0, \dots, \beta_{n-k-1})\right) = \ker \left( \alpha_0^{[\theta]} \big| \alpha_1^{[\theta]} \big| \cdots \big| \alpha_{n-k-1}^{[\theta]} \right) \] as desired. 
\end{proof}

From Proposition \ref{prop:kerSRS} it immediately follows the next result. 

\begin{corollary}\label{cor:skewRSareGab}
Let $\alpha \in E$ such that $\{ \alpha, \theta(\alpha),\ldots,\theta^{n-1}(\alpha)\}$ is a normal basis and let $\beta=\theta(\alpha)\alpha^{-1}$. Let moreover $\delta$ be an integer such that $1 \leq \delta \leq n$. Then
\[\operatorname{sRS}^\theta_{\beta}(n,\delta)=\G_{\delta-1,\theta}(\alpha^{[\theta]})^\perp.\]
In particular, skew Reed-Solomon codes are MRD codes of dimension $n-\delta+1$ and minimum rank distance $\delta$.
\end{corollary}

Let \(\alpha \in E\) such that \(\{\alpha, \theta(\alpha), \dots, \theta^{n-1}(\alpha)\}\) is a \(K\)- basis. Let \(\beta = \theta(\alpha) \alpha^{-1}\). It is well known that
\[
x^n-1 = \lclm{\{x - \theta^i(\beta) ~|~ 0 \leq i \leq n-1\}},
\]
see e.g. \cite[Theorem 5.3]{Delenclos/Leroy:2007} 

\begin{definition}\label{def:Tbetag}
Let \(g \in R\) such that \(g \mid_r x^n-1\), \(\mathcal{C} = \mathcal{R}g\) and \(\widehat{\mathcal{C}} = \mathcal{S} g\). The \emph{\(\beta\)-defining set of \(g\)} is 
\[
T_\beta(g) = \left\{ 0 \leq i \leq n-1 ~\big|~ x-\theta^i(\beta) \mid_r g \right\}.
\]
In particular, 
\(\lclm{\{x - \theta^i(\beta) ~|~ i \in T_\beta(g)\}} \mid_r g\).
\end{definition}


\section{Skew Roos bound for the Hamming metric}\label{sec:RboundHamming}

In this section, we will keep the notation of Definition \ref{def:Tbetag}. Hence we will write $\C$ for the  skew cyclic code $\C= \mathcal{R}g$, where $g \in R$ is such that $g \mid_r x^n-1$, and \(\widehat{\mathcal{C}} = \mathcal{S} g\). 

\begin{lemma}\label{chainofsubspaces}
Let \(\alpha_1, \dots, \alpha_{t+r} \in E\) be linear independent elements over \(K\). Let \(\{k_0, \dots, k_r\} \subseteq \{0, \dots, n-1\}\) be such that \(k_r - k_0 \leq t+r-1\) and \(k_{j-1} < k_{j}\) for \(1 \leq j \leq r\). Let 
\[
A_0 = \begin{pmatrix}
\theta^{k_0}(\alpha_1) & \theta^{k_1}(\alpha_1) & \cdots & \theta^{k_r}(\alpha_1) \\
\vdots & \vdots & \ddots & \vdots \\
\theta^{k_0}(\alpha_{t+r}) & \theta^{k_1}(\alpha_{t+r}) & \cdots & \theta^{k_r}(\alpha_{t+r})
\end{pmatrix}.
\]
Let \(s \in \{0, \dots, n-1\}\) such that \((s,n) = 1\) and 
\[
A_i = \left(\begin{array}{c|c|c|c} A_0 & \theta^s(A_0) & \dots & \theta^{si}(A_0) \end{array} \right).
\]
Then \(\rank(A_{t-1}) = t + r\). 
\end{lemma}

\begin{proof}
Let \(\mathcal{A}_i \subseteq E^{t+r}\) be the column space of \(A_i\), so \(\dim(\mathcal{A}_i) = \rank(A_i)\). Observe that 
\[
\{k_0, \dots, k_r\} \subseteq \{k_0, k_0 + 1, \dots, k_0 + t + r - 1\},
\]
hence \(A_0\) is obtained from 
$$A = 
\begin{pmatrix}
\theta^{k_0}(\alpha_1) & \theta^{k_0+1}(\alpha_1) & \cdots & \theta^{k_0+t+r-1}(\alpha_1) \\
\vdots & \vdots & \ddots & \vdots \\
\theta^{k_0}(\alpha_{t+r}) & \theta^{k_0+1}(\alpha_{t+r}) & \cdots & \theta^{k_0+t+r-1}(\alpha_{t+r})
\end{pmatrix}
$$
deleting some columns. By the Circulant Lemma (Lemma \ref{circulantlemma}), \(\rank(A) = t+r\), hence \(\dim(\mathcal{A}_0) = \rank(A_0) = r+1\). Assume by contradiction that \(\dim(\mathcal{A}_{t-1}) < t+r\). Since \(\mathcal{A}_{i} \subseteq \mathcal{A}_{i+1}\) for \(0 \leq i \leq t-2\), it follows that there exists \(0 \leq j \leq t-2\) such that \(\dim(\mathcal{A}_j) = \dim(\mathcal{A}_{j+1})\), i.e. \(\mathcal{A}_j = \mathcal{A}_{j+1}\). Since \(\mathcal{A}_{i+1} = \mathcal{A}_i + \theta^s(\mathcal{A}_i)\) for \(0 \leq i \leq t-2\), it follows that \(\mathcal{A}_j = \theta^s(\mathcal{A}_j)\), i.e. \(\mathcal{A}_j\) is invariant under the action of \(\theta^s\). Hence \(\mathcal{A}_j \supseteq \mathcal{A}_0 + \theta^s(\mathcal{A}_0) + \dots + \theta^{s(t+r-1)}(\mathcal{A}_0)\). In particular \(\mathcal{A}_j\) contains the columns of 
\[
A' = 
\begin{pmatrix}
\theta^{k_0}(\alpha_1) & \theta^{k_0+s}(\alpha_1) & \cdots & \theta^{k_0+s(t+r-1)}(\alpha_1) \\
\vdots & \vdots & \ddots & \vdots \\
\theta^{k_0}(\alpha_{t+r}) & \theta^{k_0+s}(\alpha_{t+r}) & \cdots & \theta^{k_0+s(t+r-1)}(\alpha_{t+r})
\end{pmatrix}.
\]
Since \((s,n) = 1\), \(K = E^\theta = E^{\theta^s}\), so again by Lemma \ref{circulantlemma}, \(\det(A') \neq 0\), and therefore \(\dim(\mathcal{A}_j) \geq t+r\). Finally \(\mathcal{A}_{t-1} \supseteq \mathcal{A}_{j}\), so we get \(t+r > \dim(\mathcal{A}_{t-1}) \geq t+r\), that is a contradiction. 
\end{proof}

\begin{theorem}[Skew Roos bound for the Hamming metric]\label{roosboundgeneral}
Let $\C$ be the skew cyclic code $\C= \mathcal{R}g$, where $g \in R$ and \(\widehat{\mathcal{C}} = \mathcal{S} g\). Let \(\alpha \in E\) such that \(\{\alpha, \theta(\alpha), \dots, \theta^{n-1}(\alpha)\}\) is a \(K\)-basis and let \(\beta = \theta(\alpha) \alpha^{-1}\). Moreover, assume that there are \(b, s, \delta, k_0, \dots, k_r\) such that, \((s,n) = 1\),  \(k_j < k_{j+1}\) for $0\leq j\leq r-1$, \(k_r - k_0 \leq \delta + r - 2\), and \(b + si + k_j \in T_\beta(g)\) for all \(0 \leq i \leq \delta-2\) and \(0 \leq j \leq r\). Then \(\distance_H(\mathcal{C}) \geq \distance_H(\widehat{\mathcal{C}}) \geq \delta + r\).
\end{theorem}

\begin{proof}
The bound \(\distance_H(\mathcal{C}) \geq \distance_H(\widehat{\mathcal{C}})\) follows since \(\mathcal{C}\) is a subfield subcode of \(\widehat{\mathcal{C}}\). 
Let \(w = \delta+r-1\) and let \(c \in \widehat{\mathcal{C}} = \mathcal{S}g\) such that \(\weight(c) \leq w\), i.e. \(c = \sum_{h=1}^w c_h x^{l_h}\) for suitable \(\{l_1, \dots, l_w\} \subseteq \{0, \dots, n-1\}\). For each \(0 \leq i \leq \delta - 2\) and \(0 \leq j \leq r\), \(x - \theta^{b + si + k_j}(\beta) \mid_r c\), so 
\[
\begin{split}
0 &= \textstyle\sum_{h=1}^w c_h \norm{l_h}{\theta^{b+si+k_j}(\beta)} \\
&= \theta^{b+si+k_j}(\alpha)^{-1} \textstyle\sum_{h=1}^w c_h {\theta^{b+si+k_j+l_h}(\alpha)}.
\end{split} 
\]
We get that \(\bar{c}:=(c_1,\ldots,c_w)\) is in the left kernel of the matrix \(\theta^{b}(B)\) where
\[
B=\left(\begin{array}{c|c|c|c}
A_0 & \theta^{s}(A_0) & \cdots & \theta^{s(\delta-2)}(A_0)
\end{array}\right)
\]
and 
\[
A_0 = \Big( \theta^{k_j+l_h}(\alpha) \Big)_{\substack{1 \leq h \leq w \\ 0 \leq j \leq r}}.
\]
Applying Lemma \ref{chainofsubspaces} with \(t = \delta-1\), we get that \(\rank(B)=w\). Hence, $\bar{c}=0$ and, so, $c=0$ is the only element in \(\mathcal{S} g\) of weight at most \(\delta+r-1\).
\end{proof}


If \(s = 1\) we obtain an instance of the BCH bound, hence we can decode. 

\begin{proposition}\label{roosbound}
Let $\C$ be the skew cyclic code $\C= \mathcal{R}g$, where $g \in R$ and \(\widehat{\mathcal{C}} = \mathcal{S} g\). Let \(\alpha \in E\) such that \(\{\alpha, \theta(\alpha), \dots, \theta^{n-1}(\alpha)\}\) is a \(K\)-basis and let \(\beta = \theta(\alpha) \alpha^{-1}\). Moreover, assume that there are \(b, \delta, k_0, \dots, k_r\) such that \(k_j < k_{j+1}\) for $0\leq j\leq r-1$, \(k_r - k_0 \leq \delta + r - 2\), and \(b + i + k_j \in T_\beta(g)\) for all \(0 \leq i \leq \delta-2\) and \(0 \leq j \leq r\). Then \(\operatorname{sRS}^\theta_{\theta^{b+k_0}(\beta)}(n,\delta+r) \supseteq \widehat{\mathcal{C}}\). In particular \(\distance_H(\mathcal{C}) \geq \delta + r\) and \(\distance_R(\mathcal{C}) \geq \delta + r\). 
\end{proposition}

\begin{proof}
Up to replacing \(\beta\) by \(\theta^{b}(\beta)\), we may assume \(b = 0\). Since \(k_j < k_{j+1}\) for all \(j\), it follows that \(k_l \geq k_j + (l-j)\) for all \(j \leq l\). Assume, for a contradiction, that \(k_j + \delta - 1 < k_{j+1}\). Then
\[
k_0 + j + \delta - 1 \leq k_j + \delta - 1 < k_{j+1} \leq k_r - (r - j - 1),
\]
and consequently
\[
\delta + r - 2 = (j + \delta - 1) + (r - j - 1) < k_r - k_0,
\]
which is incompatible with the hypothesis \(k_r - k_0 \leq \delta + r - 2\). Therefore \(k_{j+1} \leq k_j + \delta - 1\) and 
\[
\bigcup_{j=0}^r \{k_j + i ~|~ 0 \leq i \leq \delta-2\} = [k_0, k_r + \delta - 2] \cap \Zset. 
\]
Since \(k_r \geq k_0 + r\), it follows that
\[
[k_0, k_0 + \delta + r - 2] \cap \Zset \subseteq T_{\beta}(g),
\]
so, if \(f = \lclm{\{x - \theta^{k_0+i}(\beta) ~|~ 0 \leq i \leq \delta + r - 2\}}\), we have that
\[
f \mid_r g.
\]
This implies that \(\mathcal{S}f \supseteq \mathcal{S}g\). Since \(\mathcal{S}f = \operatorname{sRS}^\theta_{\theta^{k_0}(\beta)}(n,\delta+r)\) and \(\mathcal{S}g = \widehat{\mathcal{C}} \supseteq \mathcal{C}\), the result follows using Proposition \ref{prop:skewBCHbound} and Corollary \ref{cor:skewRSareGab}. 
\end{proof}

\begin{corollary}
Assume that there are \(b, \delta, s, k_0, \dots, k_r\) such that \(k_j < k_{j+1}\) for $0\leq j \leq r-1$, \(k_r - k_0 \leq \delta + r - 2\), \((s,n) = 1\), and \(b + is + k_js \in T_\beta(g)\) for all \(0 \leq i \leq \delta-2\) and \(0 \leq j \leq r\). Then there exists $b^\prime$ such that \(\operatorname{sRS}^{\theta^s}_{\theta^{b'+k_0s}(\beta)}(n,\delta+r) \supseteq \widehat{\mathcal{C}}\). In particular \(\distance_R(\mathcal{C}) \geq \delta + r\) and  \(\distance_H(\mathcal{C}) \geq \delta + r\). 
\end{corollary}

\begin{proof}
Apply Proposition \ref{roosbound} to \(\theta^s\) and \(b' = bu\) where \(us + vn = 1\). 
\end{proof}

Theorem \ref{roosboundgeneral} and Proposition \ref{roosbound} use that the corresponding distances of \(\mathcal{C}\) are bounded by the distances of \(\widehat{\mathcal{C}}\). In these cases both distances are closely related as next results show. 

Let \(\pi = \theta^\mu\). Then \(F = E^\pi\). The proof of next proposition is essentially \cite[Theorem 9]{VanLint/Wilson:1986}.

\begin{proposition}
Let \(A = \{\alpha_{1}, \dots, \alpha_k\} \subseteq E\) such that \(\pi\) induces a permutation on \(A\), 
i.e., for all \(1 \leq j \leq k\), there exists a unique \(1 \leq \pi(j) \leq k\) such that \(\pi(\alpha_j) = \alpha_{\pi(j)}\). Let \(\widehat{\mathcal{C}} = \ker \left( \alpha_1^{[\theta]} \big| \cdots \big| \alpha_{k}^{[\theta]} \right) \subseteq E^n\) and \(\mathcal{C} = \widehat{\mathcal{C}} \cap F^n\). Then \(\distance_H(\widehat{\mathcal{C})} = \distance_H(\mathcal{C})\). 
\end{proposition}

\begin{proof}
Since \(\mathcal{C} \subseteq \widehat{\mathcal{C}}\), it follows that \(\distance_H(\mathcal{C}) \geq \distance_H(\widehat{\mathcal{C}})\). Let \(c = (c_0, \dots, c_{n-1}) \in \widehat{\mathcal{C}}\) such that \(\weight(c) = \distance_H(\widehat{\mathcal{C}})\). The hypothesis \(\pi(\alpha_j) = \alpha_{\pi(j)}\) implies that \(\pi(\theta^i(\alpha_j)) = \theta^i(\alpha_{\pi(j)})\), so, for each \(0 \leq j \leq s-1\), \(\pi^j\) induces a permutation on the columns of \(\left( \alpha_1^{[\theta]} \big| \cdots \big| \alpha_{k}^{[\theta]} \right)\). Since \(c \in \widehat{\mathcal{C}}\), 
\[
(c_0, \dots, c_{n-1}) \left( \alpha_1^{[\theta]} \big| \cdots \big| \alpha_{k}^{[\theta]} \right) = 0,
\]
so 
\[
(\pi^j(c_0), \dots, \pi^j(c_{n-1})) \left( \alpha_1^{[\theta]} \big| \cdots \big| \alpha_{k}^{[\theta]} \right) = 0
\]
for each \(0 \leq j \leq s-1\). It follows that 
\[
(\trace_{E/F}(c_0), \dots, \trace_{E/F}(c_{n-1})) \left( \alpha_1^{[\theta]} \big| \cdots \big| \alpha_{k}^{[\theta]} \right)=0,
\]
i.e. \((\trace_{E/F}(c_0), \dots, \trace_{E/F}(c_{n-1})) \in \widehat{\mathcal{C}}\). Up to replacing \(c\) with some scalar multiple, we can assume \(0 \neq (\trace_{E/F}(c_0), \dots, \trace_{E/F}(c_{n-1})) \in F^n \), hence \((\trace_{E/F}(c_0), \dots, \trace_{E/F}(c_{n-1})) \in \mathcal{C} \setminus \{0\}\). Therefore
$$\distance_H(\mathcal{C}) \leq \weight(\trace_{E/F}(c_0), \dots, \trace_{E/F}(c_{n-1})) 
\leq \weight(c) = \distance_H(\widehat{\mathcal{C}}), 
$$
and then we have the equality. 
\end{proof}

By using the notation of \cite[p. 94]{Gomez/Lobillo/Navarro/Neri:2018}, let \(C_n = \{0, 1, \dots, n-1\}\) be regarded as a cyclic group of order $n$ and, since $n=\nu\mu$, \(\mu C_n = \{0, \mu, \dots, (\nu-1)\mu\}\) is a subgroup of order $\nu$ of $C_n$. Moreover, let $C_n/\mu C_n$ be the quotient group. If \(T = T^1 \cup \dots \cup T^\ell \subseteq C_n\) such that \(T^j \in C_n/\mu C_n\), it follows that \(i \in T\) implies \(i + \mu \in T\). A set with this property is said to be \emph{$\mu$-closed}. The defining set of a polynomial \(g \in R = F[x;\sigma]\) is \(\mu\)-closed because \(F = E^\pi\). 

\begin{proposition}
Let \(g \in R\) such that \(g \mid_r x^n-1\). Let \(\alpha \in E\) such that \(\{\alpha, \theta(\alpha), \dots, \theta^{n-1}(\alpha)\}\) is a normal basis. Let \(\beta = \theta(\alpha)\alpha^{-1}\) and let \(T_\beta(g) = \left\{i \in C_n ~:~ x - \theta^i(\beta) \mid_r g \right\}\). Then \(\pi\) induces a permutation on \(\left\{\theta^i(\alpha) ~|~ i \in T_\beta(g) \right\}\). 
\end{proposition}

\begin{proof}
By \cite[Lemma 4.3]{Gomez/Lobillo/Navarro/Neri:2018}, \(T_\beta(g) = T^1 \cup \dots \cup T^\ell\) for some cosets \(T^j \in C_n/\mu C_n\). Let \(A = \left\{\theta^i(\alpha) ~|~ i \in T_\beta(g) \right\}\). If \(\theta^i(\alpha) \in A\), 
\[
\pi(\theta^i(\alpha)) = \theta^{i+\mu}(\alpha) \in A
\]
because \(i \in T_\beta(g)\) implies \(i+\mu \in T_\beta(g)\). 
\end{proof}

\begin{example}\label{ex:code}
Let $F=\F_{2^6}$ be the finite field with $2^6$ elements, $a$ be a primitive element satisfying $a^6+a^4+a^3+a+1$ and consider the automorphism $\sigma:F\to F$ given by $\sigma(a)=a^2$. The order of $\sigma$ is $6$. 

Moreover, let $E=\F_{2^{12}}$ be an extension field of $F$. Let $\gamma$ be a primitive element of $E$ satisfying $\gamma^{12}+\gamma^7+\gamma^6+\gamma^5+\gamma^3+\gamma+1$. The embedding $\varphi: F\to E$ is defined as $\varphi(a)= \gamma^{9}+\gamma^5 + \gamma^4+\gamma^2+\gamma=\gamma^{65}$. Let $\theta: E \to E$ be the extension of the automorphism $\sigma$ to the field $E$, that is the Frobenius automorphism of order $12$. 

Now, fix $\alpha:=\gamma^5$ to be a normal element of $E$ as a $\F_2$-vector space. Hence $\beta := \theta(\alpha)\alpha^{-1} = \gamma^{5}$. Choose the parameters of the Roos bound as $b=0$, $\delta=3$, $r=1$, $k_0=9$ and $k_1=10$. It follows that the defining set we are looking for is  $T_\beta(g)=\{2,3,4,8,9,10\}$. Now we compute the least common left multiple $\lclm{x-\theta^i({\beta})}^{i=2,3,4,8,9,10}\in F[x;\sigma]$ which defines a skew cyclic code of dimension $6$ and distance at least $4$. In particular, the code has generator polynomial 
$$g = x^6+a^{31}x^5 + a^{26}x^4+ax^3 + a^{5}x^2 + a^{43}x + a^{49}.$$

With the aid of the software Magma \cite{MR1484478}, we can then compute the exact distance of the code that turns to be $6$. Therefore, the code $\mathcal{C} = \mathcal{R} g$ is a $[12,6,6]$ code over the field $F=\F_{2^6}$.
\end{example}

\begin{example}\label{ex:nonConsecutiveKs}
Let $K=\F_2$, $F=\F_{2^7}$, $a$ be a primitive element and $\sigma:F\to F$, given by $\sigma(a)=a^2$. Let $E = \F_{2^{14}}$ be the extension field of $F$ of degree $2$ and $\gamma$ be a primitive element of $E$. By following the Example \ref{ex:code}, let $\alpha:=\gamma^7$ be a normal element of $E$ as $K$-vector space and fix $\beta:=\theta(\alpha)\alpha^{-1} = \gamma^7$. Consider $b=0, \delta=3, r=2, k_0=2, k_1=4, k_2 =5$ as the parameters of the Roos bound. It follows that the defining set for the code we are constructing is $T_\beta(g) = \{ 2, 3, 4, 5, 6, 9, 10, 11, 12, 13 \}$, and $g$ is computed as the least common left multiple $\lclm{\{x-\theta^i(\beta) ~|~ i\in {T_\beta(g)}\}} \in F[x;\sigma]$. The code generated by $g$ is a $[14, 4, 11]$ MDS linear code over $\F_{2^7}$. 
\end{example}

\begin{example}
We are going to include an example concerning convolutional codes. Convolutional codes can be equivalently described as direct summands of \(\field{}[z]^n\), where \(\field{}\) is a finite field, or as a vector subspace of \(\field{}(z)^n\), the field of rational functions over a finite field. This equivalence was firstly established in \cite[Theorem 3]{Forney:1970}, and a more recent refinement can also be found in \cite[Proposition 1]{Gomez/Lobillo/Navarro:2019}. For this example we follow an analogous construction to \cite[Example 2.5]{Gomez/Lobillo/Navarro/Neri:2018}. Let \(F = \field[16](z)\) and \(\sigma : F \to F\) the automorphism defined by \(\sigma(z) = \frac{b^9}{z + b^4}\), where \(\field[16] = \field[2][b]/(b^4+b+1)\). This is an automorphism of order \(\mu = 15\) and, by L\"uroth's Theorem \cite[\S 10.2]{vanderWaerden:1949}, the invariant subfield is \(K = F^\sigma = \field[16](u)\) for some \(u \in \field[16](z)\). Let \(\field[256] = \field[2][a]/(a^8+a^4+a^3+a^2+1)\). It is straightforward to check that a canonical embedding \(\epsilon: \field[16] \to \field[256]\) is defined by \(\epsilon(b) = a^{17}\). Let \(\pi : \field[256] \to \field[256]\) be the automorphism defined by \(\pi(a) = a^{16}\), and let also denote by \(\pi\) the canonical extension to \(E = \field[256](z)\), i.e. 
\[
\pi \left( \frac{a_0 + a_1 t + \dots + a_m t^m}{b_0 + b_1 t + \dots + b_{m'} t^{m'}} \right) = \frac{a_0^{16} + a_1^{16} t + \dots + a_m^{16} t^m}{b_0^{16} + b_1^{16} t + \dots + b_{m'}^{16} t^{m'}}.
\]
We also use \(\sigma\) to denote its canonical extension to \(\sigma : E \to E\), so \(\sigma(z) = \frac{a^{153}}{z + a^{68}}\).  Since \(\field[16] = \field[256]^\pi\), it follows that \(\sigma \pi = \pi \sigma\), so \(\theta = \sigma \pi : E \to E\) is an extension of \(\sigma\) of degree \(\nu = 2\). In order to build a skew cyclic convolutional code of a designed Hamming distance using the Roos bound, we need a normal basis of \(E\) over \(K = E^\theta\). Such a basis can be obtained from \(\alpha = a z\), and the corresponding root is \(\beta = \theta(\alpha) \alpha^{-1} = \frac{a^{168}}{z^2 + a^{68}z}\). Let \(T = \{0, 2, 3, 4, 7, 9, 10, 11, 15, 17, 18, 19, 22, 24, 25, 26\}\). Then \(T\) is \(\mu\)-closed, and \(g = \lclm{\{x - \theta^i(\beta) ~|~ i \in T\}}\) generates a skew cyclic convolutional code of rate \(14/30\). This polynomial has degree \(16\) and its coefficients are rational functions up to degree \(11\) which we have computed with the aid of \cite{sage}. 
Since \(T \supseteq \{0, 2, 3, 4, 7, 9, 10, 11\}\), which correspond with the parameters \(b = 0\), \(\delta = 3\), \(s = 7\) and \(k_0, k_1, k_2, k_3 = 0,2,3,4\). So, its Hamming distance is bounded from below by \(\delta + r = 3 + 3 = 6\). 
\end{example}

We conclude this section including the following table, which provides a list of skew cyclic codes, computed as in Example \ref{ex:code}. Hence $F = K(a)$, where $a$ is a primitive element of $F$ and $E =K(\gamma)$, with $\gamma$ primitive element of $E$. The generator polynomials of the skew cyclic codes in the table are computed by the aid of Magma \cite{MR1484478} as least common left multiples (we omit to write it for brevity). Moreover, always with the aid of Magma, we computed the effective minimum distances of the constructed skew-cyclic codes. Observe that in some cases with this construction we obtain codes reaching the Singleton bound.

\begin{table}[!htbp]
\centering
\tiny{\begin{tabular}{|c|c|c|c|c|c|c|c|c|}
\hline
$K$ & $F$ & $E = K(\gamma)$ & $\alpha^{\phantom{G}}_{\phantom{2}}$                                                              & $b$ & $\delta$ & $r$ & $T_\beta(g)$                             & $[n,k,d]$                        \\ \hline
$\mathbb{F}_2$ & $\mathbb{F}_{2^6}$   & $\mathbb{F}_{2^{12}}$  &       $\gamma^5$                                                            & 0   & 3        & 1   & $\{2,3,4,8,9,10\}$                    & $[12,6,6]$                       \\ \hline
$\mathbb{F}_2$ & $\mathbb{F}_{2^6}$   & $\mathbb{F}_{2^{12}}$        &  $\gamma^5$                                                            & 0   & 3        & 1   & $\{1,2,3,4,7,8,9,10\}$                & $[12,4,8]$                       \\ \hline
$\mathbb{F}_2$ & $\mathbb{F}_{2^5}$   & $\mathbb{F}_{2^{20}}$        &  $\gamma^{11}$                                                         & 0   & 3        & 1   & $\{1,2,3,6,7,8,11,12, 13, 16,17,18\}$ & $[20,8,11]$                      \\ \hline
$\mathbb{F}_2$ & $\mathbb{F}_{2^7}$   & $\mathbb{F}_{2^{14}}$        &  $\gamma^{7}$                                                          & 0   & 3        & 1   & $\{0,5,6,7,12,13\}$                   & $[14,8,7]^*$  \\ \hline
$\mathbb{F}_2$ & $\mathbb{F}_{2^7}$   & $\mathbb{F}_{2^{14}}$        &  $\gamma^{7}$                                                          & 0   & 3        & 2   & $\{0,1,2,3,4,7,8,9,10,11\}$           & $[14,4,11]^*$\\ \hline
$\mathbb{F}_3$ & $\mathbb{F}_{3^6}$   & $\mathbb{F}_{3^{12}}$        &  $\gamma^{7}$                                                          & 0   & 3        & 1   & $\{2,3,4,8,9,10\}$                    & $[12,6,7]^*$ \\ \hline
$\mathbb{F}_3$ & $\mathbb{F}_{3^5}$   & $\mathbb{F}_{3^{15}}$        &  $2\gamma^{13} + \gamma^{11} +\gamma^{10}+2$                           & 0   & 3        & 1   & $\{2,3,4,7,8,9,12,13,14\}$            & $[15,6,10]^*$ \\ \hline
$\mathbb{F}_5$ & $\mathbb{F}_{5^5}$   &$\mathbb{F}_{5^{10}}$        &  $\gamma^{9} + \gamma^{7}+\gamma^{6}+3\gamma^{5}+2\gamma^{3}+\gamma+3$ & 0   & 3        & 1   & $\{0,1,2,5,6,7\}$                     & $[10,4,7]^*$                       \\ \hline
\end{tabular}}
\vspace{0.5cm}
\caption{Skew cyclic codes constructed using the Roos bound. The rows in which appears a $^*$ indicate that the corresponding code is MDS.}\label{Table}
\end{table}

\section{Skew Roos bound for the rank metric}\label{sec:RboundRank}
In this section we provide the rank metric version of the skew Roos bound, which improves the bound of Theorem \ref{roosboundgeneral}. The proof uses all the tools developed in the previous sections, and in particular it relies on Theorem \ref{roosboundgeneral}, Lemmas \ref{circulantlemma} and \ref{chainofsubspaces} and Proposition \ref{prop:minrank}.

Also in this section we will use the notation introduced in Definition \ref{def:Tbetag}, writing  $\C= \mathcal{R}g$, where $g \in R$ is such that $g \mid_r x^n-1$, and \(\widehat{\mathcal{C}} = \mathcal{S} g\). 


\begin{theorem}[Skew Roos bound for the rank metric]\label{thm:roosrank}
Assume that there are $b, s, \delta, k_0, \dots, k_r$ such that $(s,n) = 1$,  $k_j < k_{j+1}$ for $0\leq j \leq r-1$, $k_r - k_0 \leq \delta + r - 2$, and $b + si + k_j \in T_\beta(g)$ for all $0 \leq i \leq \delta-2$ and $0 \leq j \leq r$. Then  $\distance_R(\C) \geq \distance_R(\widehat{\mathcal{C}}) \geq \delta + r$
\end{theorem}

\begin{proof}
As before \(\distance_R(\mathcal{C}) \geq \distance_R(\widehat{\mathcal{C}})\) because \(\mathcal{C}\) is a subfield subcode of \(\widehat{\mathcal{C}}\).
 By Proposition \ref{prop:minrank}, we need to prove that for every $M^{-1} \in \GL_n(K)$, we have $\distance_H( \widehat{\C}\cdot M^{-1})\geq \delta+r$. Take a generic $M \in \GL_n(K)$, define $w = \delta+r-1$ and consider $c \in \widehat{\C} \cdot M^{-1}$ such that $\weight(c) \leq w$, i.e. \(c = \sum_{h=1}^w c_h x^{l_h}\) for a suitable $S:=\{l_1, \dots, l_w\} \subseteq \{0, \dots, n-1\}$. 
 Denote by $M_S$ the matrix obtained from $M$ only selecting the rows indexed by the elements in $S$ (here we assume the row indices to be $0,1,\ldots, n-1$). As in the proof of Theorem \ref{roosboundgeneral}, we get that $\bar{c}:=(c_1,\ldots, c_w)$ belongs to the left kernel of the matrix $M_S\tilde{B}$, where
 $$\tilde{B}= \left(\begin{array}{c|c|c|c}
A & \theta^{s}(A) & \cdots & \theta^{s(\delta-2)}(A)
\end{array}\right),$$
 and 
 $$A=\Big( \theta^{k_j+h}(\alpha) \Big)_{\substack{0 \leq h \leq n-1 \\ 0 \leq j \leq r}}.$$
 Now observe that
 \begin{align*}M_S\tilde{B}&=\left(\begin{array}{c|c|c|c}
M_SA & M_S \theta^{s}(A) & \cdots & M_S\theta^{s(\delta-2)}(A)
\end{array}\right) \\
&= \left(\begin{array}{c|c|c|c}
M_SA & \theta^{s}(M_SA) & \cdots & \theta^{s(\delta-2)}(M_SA)
\end{array}\right), \end{align*}
where the last equality follows from the fact that the coefficients of $M_S$ are in $K$, and hence are fixed by $\theta$. We observe now that the matrix $M_SA$ is of the form
$$M_SA=A_0 = \Big( \theta^{k_j}(\beta_h) \Big)_{\substack{1 \leq h \leq w \\ 0 \leq j \leq r}},$$
where the elements $\beta_h$'s are given by $\beta_h=M_{\{l_h\}}\alpha^{[\theta]}$, and are linearly independent over $K$. Hence, by Lemma \ref{circulantlemma}, $A_0$ has rank $r$.
At this point, applying Lemma \ref{chainofsubspaces} on the matrices $M_S\tilde{B}$ and $M_SA=A_0$, with $t = \delta-1$ we get that $\rank(M_S\tilde{B})=w$ and hence $\bar{c}=0$, so $c=0$ is the only element in $\widehat{\C} \cdot M^{-1}$ of weight at most $\delta + r-1$. This proves that $\distance_H(\widehat{\C} \cdot M^{-1})\geq \delta+r$ and concludes the proof.
\end{proof}

In Section \ref{sec:rankmetric}, we mentioned that Gabidulin codes are MRD codes, since their parameters attain a Singleton-like bound for the rank metric. Actually, there are two Singleton-like bounds for the rank metric, depending on how the length and the extension degree of the code are related. Formally, let $\C$ be an $[n,k,d]_{F/K}$ rank-metric code and let $\mu=[F:K]$, then 
\begin{align}
    k & \leq n-d+1  \label{eq:Sing1}\\
    k & \leq \frac{n}{\mu}(\mu- d+1) \label{eq:Sing2}  
\end{align} 
In particular, one considers  inequality \eqref{eq:Sing1} when $n\leq \mu$, and inequality \eqref{eq:Sing2} if $\mu$ divides $n$. 
In this setting, an $[n,k,d]_{F/K}$ rank-metric code is \emph{maximum rank distance (MRD)} if its parameters meet with equality  one of the two bounds above.

 Since in the construction of rank-metric codes that  we gave using the skew Roos bound of Theorem \ref{thm:roosrank} we deal with $n=\mu \nu$, we should only consider inequality \eqref{eq:Sing2}, that with our notation becomes
 \begin{equation}\label{eq:Sing3}
 k \leq \nu(\mu-d+1).
 \end{equation}
 Hence, a code $\C$ satisfying the hypotheses of Theorem \ref{thm:roosrank} is an $[n,k,\geq \delta+r]_{F/K}$ rank-metric code, where $k=n-\deg g \leq \nu(\mu-\delta-r+1)$. 
 
 \begin{example}\label{ex:codeRank}
 Consider the code $\C$ constructed in Example \ref{ex:code} endowed with the rank metric. Putting together the Singleton-like bound in \eqref{eq:Sing3} and the skew Roos bound for the rank metric of Theorem \ref{thm:roosrank}, we get that $\C$ is a $[12,6,\geq 4]_{F/K}$ rank-metric code, where $F=\mathbb{F}_{2^6}$ and $K=\mathbb{F}_2$, which satisfies the following chain of inequalities
 $$ 4=\delta+r\leq \distance_R(\C) \leq \mu-\frac{\mu k}{n}+1=4.$$
 Therefore, the inequalities above are all equalities and $\C$ is an MRD code.
 \end{example}
 
 \begin{example}
 Consider now the code $\C$ constructed in Example \ref{ex:nonConsecutiveKs} equipped with the rank metric. In this case, combining the Singleton-like bound in \eqref{eq:Sing3} with the skew Roos bound for the rank metric of Theorem \ref{thm:roosrank}, we deduce that $\C$ is a $[14,4]_{F/K}$ code with $ F=\mathbb{F}_{2^7}$, $K=\mathbb{F}_2$ and whose minimum rank distance satisfies 
 $$ 5=\delta+r\leq \distance_R(\C) \leq \mu-\frac{\mu k}{n}+1=6.$$
 Hence, according to the two bounds, we  have an MRD code or an almost MRD code (i.e. $\distance_R(\C)=\mu-\frac{\mu k}{n}$), depending on the exact value of $ \distance_R(\C)$. However, studying the set $T_\beta(g)$ more carefully, we can see that it also satisfy a skew Roos bound with $b=0$, $s=1$, $\delta'=6$ and $r=0$ (i.e. a skew BCH bound). Hence the code $\C$ is actually an MRD code.
 \end{example}
 
 \begin{remark}
 It is very interesting to observe that the skew-cyclic code $\C$ considered in  Examples \ref{ex:code} and  \ref{ex:codeRank} is not an MDS code, but it is an MRD code (with respect to the Singleton-like bound in \eqref{eq:Sing2}).  This is quite surprising since for $[n,k]_{F/K}$ rank-metric codes such that $n\leq [F:K]$, i.e. when we need to consider the Singleton-like bound in \eqref{eq:Sing1}, MRD codes are also MDS. In addition, we have by construction that $\C=\widehat{\C}\cap F^n$, i.e. $\C$ is a subfield subcode of a rank-metric code $\widehat{\C} \leq E^n$. It is possible to verify that $\widehat{\C}$ is not an MRD code (since it has codewords of rank weight equal to $6$), even though $\C$ is MRD.
 \end{remark}
 
 In the following table, we analyze the same skew cyclic codes from Table \ref{Table}, endowed with the rank metric. Observe that, in all the cases, we get almost MRD codes or MRD codes.
 
 \begin{table}[!htbp]
\centering
\begin{tabular}{|c|c|c|c|c|c|c|c|c|}
\hline
$K$ & $F$ & $E$ & $\delta$ & $r$ & $n$ & $k$ & $\mu-\frac{\mu k}{n}+1$ & $\distance_R$                    \\ \hline
$\mathbb{F}_2$ & $\mathbb{F}_{2^6}$   & $\mathbb{F}_{2^{12}}$  &   $3$ & $1$ & 12   & $6$ & $4$ & $4^*$                      \\ \hline

$\mathbb{F}_2$ & $\mathbb{F}_{2^6}$   & $\mathbb{F}_{2^{12}}$  &   $3$ & $1$ & 12   & $4$ & $5$ & $5^*$                      \\ \hline

$\mathbb{F}_2$ & $\mathbb{F}_{2^5}$   & $\mathbb{F}_{2^{20}}$  &   $3$ & $1$ & 20   & $8$ & $4$ & $4^*$                      \\ \hline

$\mathbb{F}_2$ & $\mathbb{F}_{2^7}$   & $\mathbb{F}_{2^{14}}$  &   $3$ & $1$  & 14   & $8$ & $4$ & $4^*$                      \\ \hline

$\mathbb{F}_2$ & $\mathbb{F}_{2^7}$   & $\mathbb{F}_{2^{14}}$  &   $3$ & $2$ & 14   & $4$ & $6$ & $6^*$                      \\ \hline

$\mathbb{F}_3$ & $\mathbb{F}_{3^{6}}$   & $\mathbb{F}_{3^{12}}$  &   $3$ & $1$ & 12   & $6$ & $4$ & $4^*$                      \\ \hline

$\mathbb{F}_3$ & $\mathbb{F}_{3^5}$   & $\mathbb{F}_{3^{15}}$  &   $3$ & $1$ & 15   & $6$ & $5$ & $4\leq \distance_R\leq 5$                      \\ \hline

$\mathbb{F}_5$ & $\mathbb{F}_{5^{5}}$   & $\mathbb{F}_{5^{10}}$  &   $3$ & $1$ & 10   & $4$ & $4$ & $4^*$                      \\ \hline
\end{tabular}
\vspace{0.5cm}
\caption{Skew cyclic rank-metric codes constructed using the Roos bound. The rows in which appears a $^*$ indicate that the corresponding code is MRD.}\label{Table:rankmetric}
\end{table}

 The behaviour of the codes constructed with respect to the rank metric can be partially understood as follows. Let $T\subseteq C_n$ be a $\mu$-closed set, i.e. such that $i \in T $ if and only if $i+\mu \in T$. This means that $T=T_{\beta}(g)$ for some $g\in \mathcal R$ and $T=T^1\cup \cdots\cup T^{\ell}$, where $T^j \in C_n/\mu C_{n}$. Hence, we can just consider for each $T^j$ a representative $i_j$ belonging to $C_{\mu}=\{0,1,\ldots,\mu-1\}$. We denote this set by $T_\beta^{F}(g):=\{i_1,\ldots,i_{\ell}\}$. 
 
 \begin{proposition}\label{prop:MRDconstruction}
 Suppose that the defining set $T_{\beta}(g)$ satisfies a skew Roos bound as in Theorem \ref{thm:roosrank} for some $\delta \geq 2$ and $r\geq 0$. Then the minimum rank distance of the code  $\C=\mathcal R g$ satisfies
 $\delta+r\leq \distance_R(\C)\leq |T_\beta^F(g)|+1$. In particular, if $|T_\beta^F(g)|=\delta+r-1$, then  $\C$ is an MRD code.
 \end{proposition}
 
 \begin{proof}
 The first inequality is the skew Roos bound of Theorem \ref{thm:roosrank}. For the second inequality, we have that $T_\beta^F(g)$ is a system of representative for $T_\beta(g)$, which is its $\mu$-closure. Therefore, $|T_{\beta}(g)|= \nu |T_{\beta}^F(g)|$ and $k=n-|T_{\beta}(g)|=\nu\mu-\nu|T_{\beta}^F(g)|$. Combining this equality with \eqref{eq:Sing3}, we obtain
 $$k =\nu\mu-\nu|T_{\beta}^F(g)| \leq \nu(\mu-d+1),$$
 from which we derive the desired inequality. The second statement follows directly.
 \end{proof}
 
 \begin{remark}
 Proposition \ref{prop:MRDconstruction} translates the skew Roos bound and the Singleton-like bound in an arithmetic problem. Indeed, it essentially requires to find a defining set with a suitable cardinality and only working modulo $n$ and $\mu$ to construct rank-metric codes whose minimum distance is upper and lower-bounded. 
 \end{remark}
 
 We can observe that in almost all the cases of Table \ref{Table:rankmetric} with $r=1$, we get $|T_\beta^F(g)|=\delta+r-1$, with the $\delta$ and the $r$ provided. In the codes from the second and the fifth rows, we get $|T_\beta^F(g)|=\delta'+r'-1$, with some different $\delta'$ and $r'$ for which $T_\beta(g)$ satisfies the skew Roos bound. 
 
 \begin{corollary}\label{cor:repeatedGab}
  Let $b,\delta', \mu, \nu, n, s$ be nonnegative integers such that $\mu,\nu \geq 1$, $2\leq \delta' \leq \mu$, $n=\mu\nu$ and $(s,n)=1$. Define $T:=\{b,b+s, b+2s, \ldots, b+(\delta'-2)s\}\subseteq C_{\mu}$, where all the elements are taken modulo $\mu$, and let $\bar{T}$ be its $\mu$-closure in $C_{n}$. Then $\bar{T}=T_\beta(g)$ for some polynomial $g \in R = F[x;\sigma]$, such that the code $\mathcal Rg$ is an $[n,n-\nu(\delta'-1),\delta']_{F/K}$ MRD code.
 \end{corollary}
 
 \begin{proof}
 First, observe that $|T|=\delta'-1$, i.e. all the elements $b+is \mod \mu$ are distinct, for $0\leq i \leq \delta'-2$. Indeed, if there are $0\leq i\leq j\leq \delta'-2$ such that $b+is\equiv b+js \mod \mu$, then we would have $(j-i)s \equiv 0 \mod \mu$. Since $(s,\mu)=1$, this implies $(j-i)\equiv 0 \mod \mu$, which implies $i-j=0$, due to the assumptions that $0\leq j-i \leq \delta'-2 \leq \mu-2$. It is left to show that the $\mu$-closure of $T$, that is $\bar{T}$, satisfies a skew Roos bound with $\delta=\delta'$ and $r=0$. However, this is clear by construction, since for every $0\leq i \leq \delta'-2$ the equivalence class of $b+is$ in $C_n/\mu C_n$ is contained in $\bar{T}$. In particular the set $\{b+is \mid 0\leq i \leq \delta'-2\} \subseteq \bar{T}$ in $C_n$. We conclude the proof using Proposition \ref{prop:MRDconstruction}.
 \end{proof}
 
 \begin{example}
 Let us fix any triple of fields $K\subseteq F \subseteq E$ such that $[F:K]=\mu=11$, and $[E:F]=\nu=7$, and take the polynomial $g$ such that $T_\beta^F(g)=\{0,1,2,3,5,6\}$. We can observe that the set $T_\beta(g)$ satisfies the skew Roos bound of Theorem \ref{thm:roosrank} with $b=0$, $s=12$, $\delta=3$, $r=3$,  $k_0=0$, $k_1=1$, $k_2=2$ and  $k_3=5$. Hence, $\delta+r=6$ and $|T_\beta^F(g)|=6$, and by Proposition \ref{prop:MRDconstruction} the code $\C=\mathcal R g$ is a $[77,49, \distance_R(\C)]_{F/K}$ rank-metric code whose minimum distance satisfies $6 \leq \distance_R(\C)\leq 7$.  
 \end{example}
 
 At this point it is important to remark that in all the construction of MRD codes of Table \ref{Table:rankmetric}, the codes satisfy also a skew Roos bound with $r=0$ and $\delta=|T_\beta^F(g)|+1$, that is they can be obtained using Corollary \ref{cor:repeatedGab}. Unfortunately, it does not seem trivial to construct MRD codes according to Proposition \ref{prop:MRDconstruction}, different from the ones in Corollary \ref{cor:repeatedGab}. Indeed, this is not possible when $\mu$ is a prime number, as shown in the following result.

\begin{proposition}\label{prop:muprime}
Let $s,b, \delta,k_0,\dots,k_r$ be integers such that $(s,n)=1$,  $k_j < k_{j+1}$ for $0\leq j \leq r-1$, $k_r - k_0 \leq \delta + r - 2$, and $b + si + k_j \in T_\beta(g)$ for all $0 \leq i \leq \delta-2$ and $0 \leq j \leq r$. Moreover, assume that $\mu$ is a prime number.
If $|T_\beta^F(g)|= \delta + r -1$, then $T_\beta(g)$ satisfies a BCH bound with $\delta^\prime = \delta +r$.
\end{proposition}

\begin{proof}
Up to replacing $\beta$ with $\theta^b(\beta)$, it is enough to prove the statement when $b=0$. Let $A:=\{k_0,\dots k_r\}\subseteq C_{\mu}$ and $B:=\{0, s, \dots s(\delta-2)\}\subseteq C_{\mu}$. In this setting we have $T_\beta^F(g)  \supseteq A+B$, where 
$$A+B = \{a+b \mid a \in A, b \in B\}$$ and all the elements are taken modulo $\mu$. First, we can suppose $\delta+r-1<\mu$, otherwise we would get a trivial code.  Moreover, we can also assume that $\delta>2$ and $r\geq 1$, otherwise we have already a BCH bound. 
Combining the hypotheses and Cauchy-Davenport Theorem \cite{cauchy1813recherches,davenport1935addition}, we have the following equalities
$$|T_\beta^F(g)|=|A+B|=|A|+|B|-1=\delta+r-1.$$
The pairs of sets $(A, B)$ for which equality holds in the Cauchy-Davenport Theorem have been characterized by Vosper in \cite{vosper1956critical}.
Applying this result in our setting, i.e. when $|A|=\delta-1>1$, $|B|=r+1>1$ and $|A|+|B|-1<\mu$, we get that $|A+B| = |A|+|B|-1$ if and only if $A$ and $B$ are representable as arithmetic progressions with the same common difference $s^\prime$ and clearly $s^\prime$ is coprime to $\mu$. Hence also $A+B=T_\beta^F(g)$ is representable as an arithmetic progression with difference $s^\prime$, and this implies that $T_\beta(g)$ satisfies a BCH bound.
\end{proof}

\section{Conclusions and open problems}\label{sec:conclusion}
In this paper, we provided a generalization of the Roos bound for skew cyclic codes in the Hamming and rank metric over a general field. The only requirement that we ask is to have a cyclic Galois extension of finite degree, but we do not require to work on finite fields. For the rank metric case, we also provide in Proposition \ref{prop:MRDconstruction} a way to arithmetically construct codes with a prescribed minimum rank distance, using the skew Roos bound of Theorem \ref{thm:roosrank}. Finally, we constructed some example of MDS codes and MRD codes over finite fields obtained using the skew Roos bounds of Theorems \ref{roosboundgeneral}  and  \ref{thm:roosrank}. 

In the second part of Proposition \ref{prop:MRDconstruction}, we suggest a way to construct MRD codes only using an arithmetic argument modulo $\mu$ and $n$. However, we could not come up with a general construction of MRD codes based on that, except for codes satisfying a skew Roos bound with parameters $\delta=|T_\beta^F(g)|+1$ and $r=0$. Hence we suggest the following open problem.

\begin{problem}
 Is it possible to give a different systematic construction of MRD codes meeting \eqref{eq:Sing2} based on Proposition \ref{prop:MRDconstruction} that can not be obtained using Corollary \ref{cor:repeatedGab}, i.e. not satisfying any skew Roos bound with parameters $\delta=|T_\beta^F(g)|+1$ and $r=0$?
\end{problem}

As shown in Proposition \ref{prop:muprime}, the answer to this question is negative when $\mu$ is prime. However, the general case is still unclear. 

\bibliographystyle{abbrv}     
\bibliography{references}   

\end{document}